\documentclass[12pt]{article}
\usepackage[margin=1in]{geometry}
\usepackage{setspace}
\usepackage{newtxtext}

\usepackage[utf8]{inputenc}

\usepackage{graphicx}
\usepackage{multirow}
\usepackage{amsmath,amssymb,amsfonts}
\usepackage{amsthm}
\usepackage{mathrsfs}
\usepackage[title]{appendix}
\usepackage{xcolor}
\usepackage{textcomp}
\usepackage{manyfoot}
\usepackage{booktabs}
\usepackage{algorithm}
\usepackage{algorithmicx}
\usepackage{algpseudocode}
\usepackage{listings}
\usepackage{enumitem}
\usepackage{array}
\usepackage{caption}
\usepackage{subcaption}
\usepackage{float}
\usepackage{url}

\allowdisplaybreaks[4]
\numberwithin{equation}{section}

\newtheorem{theorem}{Theorem}
\newtheorem{lemma}{Lemma}
\newtheorem{remark}{Remark}

\theoremstyle{definition}
\newtheorem{definition}{Definition}

\begin{document}
	
	
\begin{flushleft}
	{\Large\bfseries Uniform Strong Convergence of Kernel-Type Estimators for $m$-LENQD Sequences with Applications}
\end{flushleft}
	
\bigskip
	
\begin{flushleft}
	{\bfseries Panpan Ba$^{\rm a}$, Zhicheng Feng$^{\rm a}$, Jianhua Shi$^{\rm a,b,*}$, Yangyang Lyu$^{\rm a}$}
		
	\medskip
		
	{\small\itshape
			$^{\rm a}$School of Mathematics and Statistics, Minnan Normal University, Zhangzhou 363000, China
		}
		
		
	{\small\itshape
			$^{\rm b}$Fujian Key Laboratory of Data Science and Statistics, Zhangzhou 363000, China
		}
		
		
	\medskip
		
		{\small $^*$Corresponding author. Email: \textsf{v0085@126.com}}
\end{flushleft}
	
\bigskip
	
	
\noindent{\bfseries Abstract}
	
In this paper, we first derive the uniform strong consistency and the corresponding convergence rates for kernel estimators of both the distribution function and the probability density function for sequences of $m$-linearly extended negative quadrant dependent ($m$-LENQD) random variables under mild regularity conditions. The obtained results significantly improve and extend existing ones in the literature. As their theoretical applications, we further establish the uniform strong consistency and convergence rates for estimators of the hazard rate function and the mode of the density function. Numerical simulations and a real data analysis are conducted to assess the finite-sample performance of the proposed estimators, confirming the asymptotic properties derived in the theoretical analysis.
	
\medskip
	
	
	
	
\section{Introduction}\label{sec1}
The estimation of distribution functions and probability density functions has long been a fundamental problem in statistical inference. With the increasing complexity of data, many nonparametric estimation methods have gained significant attention due to their inherent flexibility and minimal reliance on model assumptions. As a versatile and smooth nonparametric tool, kernel estimation has been widely applied to estimate distribution and density functions.  Compared with traditional parametric models, kernel methods demonstrate superior adaptability to a wide range of practical data distributions.

Throughout this paper, let $\{X_n, n \geq 1\}$ be a sequence of identically distributed random variables with a common, but unknown, marginal probability density function (p.d.f.)$f(x)$ and corresponding distribution function (d.f.) $F(x)$. For its observation sample $X_1, \dots, X_n$, its empirical d.f. $F_n(x)$ is commonly used as the nonparametric estimation of $F(x)$, which is defined as follows
\begin{equation}
	F_n(x) = \frac{1}{n} \sum_{i=1}^{n} I(X_i \leq x), x \in \mathbb{R}.\label{eq-1}
\end{equation}

In fact, when the p.d.f. $f(x)$ exists, the smoothness of the d.f. $F(x)$ is not captured by using the formula \eqref{eq-1}. To address this issue, Yamato (1973) proposed the following kernel distribution estimator
\begin{equation}
	\widehat{F}_n(x)=\int_{-\infty}^{\infty} K\left(\frac{x - s}{h_n}\right) \,\mathrm{d}F_n(s)=\frac{1}{n} \sum_{i=1}^{n} K\left(\frac{x - X_i}{h_n}\right), x \in \mathbb{R},\label{eqn-2}
\end{equation}
where $K(x)$ is a known d.f. and $h_n > 0$ is called the bandwidth satisfying $h_n \to 0$ as $n \to \infty$. Further, if $K(x)$ is differentiable and $K'(x)=k(x)$ is its corresponding p.d.f., the kernel density estimator can be expressed as
\begin{equation}
	\hat{f}_n(x) = \widehat{F}_n'(x) = \frac{1}{n h_n} \sum_{i=1}^{n} k\left(\frac{x - X_i}{h_n}\right),x \in \mathbb{R}.\label{eq-2}
\end{equation}

Many scholars have obtained fruitful results on the asymptotic properties of kernel distribution estimators under dependent samples. For example,  Cai and Roussas (1992) and Cai (1993) investigated the uniform strong consistency of kernel distribution estimators under $\alpha$-mixing and $\varphi$-mixing conditions, respectively. Roussas (1995) studied the asymptotic normality of kernel distribution functions for NA random variables. In particular, Yang (2003) established the uniform strong convergence rate for the empirical distribution estimator under negatively association (NA) samples, and Cui et al. (2012) extended this result to linearly negative quadrant dependent (LNQD) sequences. More recently, Wu et al. (2025) extended the study of the empirical distribution function in Yang (2003) to the kernel distribution estimator $\widehat{F}_n(x)$ under asymptotically almost negatively associated (AANA) random variables, and derived its uniform strong convergence rate. Motivated by these works, this paper further extends earlier results concerning the empirical distribution estimator $F_n(x)$ to the smooth kernel distribution estimator $\widehat{F}_n(x)$ in the setting of  $m$-LENQD sequences. Moreover, under the considered setting, the convergence rate obtained in this paper improves upon that reported in Wu et al. (2025).

For kernel density estimator $\widehat {f}_n (x) $, Ling et al. (2008) systematically studied the strong consistency and uniform strong consistency of pairwise negatively quadrant dependent (pairwise NQD) sequences, but the theoretical convergence rates were not thoroughly explored. Wang et al. (2022) discussed the strong consistency under widely orthant dependent (WOD) sequences and verified its convergence performance via numerical simulations. Wang, S. et al. (2025) established the uniform complete consistency and convergence in mean of order $r$ for kernel-type density estimators under $m$-WOD samples, but they did not specify the corresponding convergence rates. 
Although  Li and Yang (2005), Li (2017) and Li et al. (2018)  obtained strong convergence rates of density estimates under NA, extended negatively dependent (END), and WOD sequences, respectively, their derivations relied heavily on the monotonicity assumption of kernel functions. To overcome these limitations, this paper constructs a kernel estimator in integral form under the more general dependent framework of $m$-LENQD sequences, and derives the uniform strong convergence rate of the kernel density estimator, thereby extending and improving the existing theoretical literature.


In the following, we recall some concepts of dependence structures. The concept of LENQD is introduced by Li (2023) as follows.
\begin{definition}
Two random variables $X_1$ and $X_2$ are said to be extended negative quadrant dependent (ENQD), if there exists a dominating constant $M\ge 1$, such that both inequalities
\[
P(X_1 \le x_1,X_2 \le x_2)\le M P(X_1 \le x_1)P(X_2 \le x_2)
\]
and
\[
P(X_1 > x_1,X_2 > x_2)\le M P(X_1 > x_1)P(X_2 > x_2)
\]
hold for all $x_1,x_2\in\mathbb{R}$. A sequence $\{X_n, n\ge 1\}$ of random variables is called LENQD with dominating constant $M\ge 1$, if for any disjoint subsets $A,B\subseteq \mathbb{N}$ and any collection $\{s_1,s_2,\dots \}$ of real numbers that are either all positive or all negative, the pair
\begin{equation*}
	\biggl(\sum_{i\in A} s_i X_i,\; \sum_{j\in B} s_j X_j\biggr)
	\text{ is ENQD with the same dominating constant } M.
\end{equation*}
\end{definition}

From Examples 1.1 and 1.2 and Remark 1.1 in Li et al. (2023), it is easily seen that the LENQD random variables is a more general dependence structure, with independent, NA, and LNQD sequences as special cases. Moreover, Li et al. (2024)   investigated the strong consistency and strong representation of the Kaplan–Meier and hazard rate estimators under censored LENQD survival data. Motivated by these previous studies on dependent random variables, in particular Li et al. (2023), El Alem (2026) introduced the concept of $m$-LENQD random variables as follows.

\begin{definition}
	Let $m\ge 1$ be a fixed integer. A sequence  of random variables \(\{X_n, n\ge 1\}\) is said to be $m$-LENQD if for any $n\ge 2$ and any $i_1,i_2,\dots,i_n$ such that $|i_k-i_j|\ge m$ for all $1\le k\neq j\le n$, the sequence $X_{i_1},X_{i_2},\dots,X_{i_n}$ is LENQD.
\end{definition}

When $m=1$, $m$-LENQD sequences reduces to the usual LENQD sequences. 
Furthermore, it is well known that if for any $n \geq 2$ and any indices $i_1,i_2,\dots,i_n$ such that $|i_k - i_j| \geq m$ for all $1 \leq k \neq j \leq n$, the random variables $X_{i_1},X_{i_2},\dots,X_{i_n}$ are independent, then we say that $\{X_n, n \geq 1\}$ exhibits $m$-dependence. Hence, the class of $m$-LENQD sequences includes NA, LNQD, LENQD, $m$-NA, $m$-LNQD, and $m$-dependence. So investigating the basic properties and limiting behavior of $m$-LENQD random variables is of great interest.

In this paper, we firstly derive the uniform strong convergence rates for the kernel distribution estimator $\hat{F}_n(x)$ and the kernel density estimator $\hat{f}_n(x)$ under a set of general regular conditions. Furthermore, considering that the risk rate function $r(x)$ (see(\ref{eq:hazard_rate_def})) has important application value in reliability theory and biomedicine, and that the mode $\theta$ (see(\ref{eqn-1})) of a population, as a robust location parameter resistant to outliers, plays a key role in characterizing data typicality and cluster analysis, we further apply our theoretical results to the estimation of these important statistical functionals.
	
\section{Assumptions and main results}\label{sec2}
For convenience, all limit relations are expressed as $n\to\infty$, unless otherwise specified. Let the symbol $C$ stand for a positive constant not depending on $n$. For two positive functions $u(n)$ and $v(n)$, we denote $u(n)=O(v(n))$ if  $\limsup_{n\to\infty}u(n)/v(n)<\infty$.  Now, we introduce a set of assumptions to prove our main results.

\begin{enumerate}[label=(H\arabic*)]
	\item 
	\begin{enumerate}[label=(\roman*)]
		\item $\int_{\mathbb{R}} u k(u) du = 0$ and $\int_{\mathbb{R}} u^2 k(u) du < \infty$.
		\item The first-order derivative $k'(\cdot)$ exists and satisfies $\int_{\mathbb{R}} |k'(u)| du < \infty$.
	\end{enumerate}
	\item The bandwidth $h_n > 0$ satisfies $h_n \to 0$ and $n^{-1} h_n^{-2} \log n \to 0$ as $n \to \infty$.
	\item $f(x)$ is twice differentiable on $\mathbb{R}$, with $f^{''}(\theta) \neq 0$, $\sup_{x \in \mathbb{R}} |f^{'}(x)| < \infty$, and $\sup_{x \in \mathbb{R}} |f^{''}(x)| < \infty$.
	\item The mode $\theta$ satisfies that for any $\varepsilon > 0$ and $x \in \mathbb{R}$, there exists $\gamma > 0$ such that $|x - \theta| \geq \varepsilon$ implies that $|f(\theta) - f(x)| \geq \gamma$.
\end{enumerate}


\begin{remark}\label{rem-2}
	The above assumptions are fairly general. Assumptions (H1) are standard regularity conditions for the kernel function $K(x)$ and are satisfied by many common kernels, including the Gaussian, Epanechnikov, quartic, and triweight kernels; see, for example, Wu et al. (2025). Assumption (H2) is needed for our results and holds for bandwidth sequences such as $h_n=n^{-\alpha}$ with $0<\alpha<1/2$. Assumption (H3) is a technical condition on the density function $f(x)$ and has also been used in Liu and Yang (2008), Wang et al. (2022), and Wang et al. (2025). Finally, Assumption (H4) is standard in kernel mode estimation; it guarantees that the mode $\theta$ is uniquely defined and consistently estimable. Similar assumptions are widely used in the literature; see, for example, Khardani et al. (2012), Wang, S. et al. (2025), and El Alem et al. (2026).
\end{remark}

In the following, we establish the uniform strong convergence results for kernel‑type distribution and density estimators under $m$-LENQD dependence. Detailed proofs are given in the Appendix.
\begin{theorem}\label{th:1'}
	Let $\{X_n, n \geq 1\}$ be a sequence of $m$-LENQD random variables, and let (H1)–(H3) hold true, then
	\begin{equation}\label{eqa}
		\sup_{x \in \mathbb{R}} \left| \widehat{F}_n(x) - F(x) \right| = O\left(\sqrt{\frac{\log n}{n}}\right)+ O(h_n^2), \quad \text{a.s.}
	\end{equation}
\end{theorem}

\begin{remark}
	Compared with the results of Yang (2003) and Cui et al. (2012), Theorem \ref{th:1'} provides the following improvements and extensions:
	\begin{enumerate}[label=(\roman*)]
		\item  The dependence structure of the sample is extended from NA and LNQD sequences to  $m$-LENQD sequences,which is much more general; 
		\item The classical empirical distribution function $F_n(x)$ is generalized to the smoother kernel distribution estimator $\widehat{F}_n(x)$;
		\item  With $h_n = n^{-1/4} \log^{1/4} n$, the convergence rate in Theorem \ref{th:1'} is sharper than $O(n^{-1/2} (\log^{1/2} n) \log \log n)$ obtained in Yang (2003) and Cui et al. (2012).
	\end{enumerate}
\end{remark}

\begin{theorem}\label{th:1}
	Let $\{X_n, n \geq 1\}$ be a sequence of $m$-LENQD random variables, and let (H1)–(H3) hold true. If $n^{-1} h_n^{-2} \log n \to 0$, as $n \to \infty$, then  
	\begin{equation}\label{eqn-8}  
		\sup_{x \in \mathbb{R}} |\hat{f}_n(x) - f(x)| = O\left(\sqrt{\frac{\log n}{nh_n^2}}\right)+ O(h_n^2), \quad \text{a.s.}   
	\end{equation}
\end{theorem}

\begin{remark}\label{re:3}
	By taking $h_n = n^{-1/6} \log^{1/6} n$, we can easily get the uniformly strong convergence rate as
	\begin{equation}
		\sup_{x \in \mathbb{R}} |\hat{f}_n(x) - f(x)| = O(n^{-1/3} \log^{1/3}n), \quad \text{a.s.}\label{eqn-9} 
	\end{equation}
	
 Wang et al. (2022) established the strong consistency of kernel density estimators for WOD sequences. In contrast, we obtain uniform strong consistency and extend the convergence domain from compact intervals $[a,b]$ to $\mathbb{R}$. Since $m$-LENQD and WOD sequences do not imply each other, our results further extend Wang et al. (2022). Moreover, although Assumption (H3) is stronger than condition (A1)(iii) of Wu et al. (2025) under AANA samples, we obtain a slightly sharper convergence rate. 

\end{remark}

As direct applications of Theorem \ref{th:1'} and Theorem \ref{th:1}, we further investigate the convergence properties of the kernel-type estimators of the hazard rate function and the mode under $m$-LENQD sequences.

\section{Kernel-type estimators for the hazard rate and mode}
\subsection{Kernel estimator for the hazard rate}
The hazard rate function is a fundamental quantity for characterizing the survival time of a random variable or the intensity of event occurrences. It measures the conditional probability that an individual who has survived up to time $x$ experiences a specific irreversible event (e.g., failure, disease onset) instantaneously thereafter. Owing to its wide applicability in reliability engineering, medical statistics, and related fields, estimation of the hazard rate function $r(x)$ is of significant theoretical and practical importance. Its explicit expression is given as follows:

\begin{equation} 
	r(x) = \frac{f(x)}{1 - F(x)}.\label{eq:hazard_rate_def}
\end{equation}

Based on \eqref{eq:hazard_rate_def}, an intuitive estimation method is to directly use the histogram density estimator $f_n (x) $ and the empirical d.f. $F_n (x) $ to approximate $f (x) $ and $F (x) $, which yields
\begin{equation} \label{eq:hazard_rate_estimators}
	r_n(x) = \frac{f_n(x)}{1 - F_n(x)},
\end{equation}
where \(F_n(x)\) is given in \eqref{eq-1}. For a given  $h_n$, the histogram density estimator $f_n(x)$  is defined as
\[
f_n(x) = \frac{1}{2n h_n} \sum_{i = 1}^{n} I(x - h_n < X_i \leq x + h_n).
\]

However, both $f_n(x)$ and $F_n(x)$ are piecewise functions, leading to a non‑smooth hazard rate estimator $r_n(x)$. To address this, existing studies often replace $f_n(x)$ with its kernel counterpart $\hat{f}_n(x)$ while retaining $F_n(x)$. For more details, see Li and Yang (2005), Li (2017), Li et al. (2018), etc. Since practical applications typically require a smoother hazard rate function, we further apply kernel smoothing for both $f(x)$ and $F(x)$, thereby obtaining a fully kernel‑based estimator of the hazard rate function
\begin{equation} 
	\hat{r}_n(x) = \frac{\hat{f}_n(x)}{1 - \widehat{F}_n(x)},\label{eq:4.2}
\end{equation}
where $\widehat{F}_n(x)$ and $\hat{f}_n(x)$ are given by \eqref{eqn-2} and \eqref{eq-2}, respectively.

Therefore, by Theorem \ref{th:1'} and Theorem \ref{th:1}, the uniform strong consistency rate for kernel estimator $\hat{r}_n$ is obtained as follows
\begin{theorem}\label{th:2}
	Let $\{X_n, n \geq 1\}$ be a sequence of $m$-LENQD random variables, and assume that the conditions of Theorem \ref{th:2} hold. Let $\bar{x} = \sup\{x \in R : F(x) < 1\}$, and the interval $[c, d]$ is any subset of $(-\infty, \bar{x})$. Then
	
	\begin{equation}\label{eqn-10} 
		\sup_{x \in [c,d]} |\hat{r}_n(x) - r(x)| = O\left(\sqrt{\frac{\log n}{nh_n^2}}\right)+ O(h_n^2), \quad \text{a.s.}
	\end{equation}	
\end{theorem}

Similar to Remark \ref{re:3}, by taking $h_n = n^{-1/6} \log^{1/6} n$ in Theorem \ref{th:2}, we obtain  

\begin{equation}
	\sup_{x \in [c,d]} |\hat{r}_n(x) - r(x)| = O(n^{-1/3} \log^{1/3} n), \quad \text{a.s.}\label{eqn-11}
\end{equation}

\begin{remark}
	Based on $\alpha$-mixing samples,  Cai and Roussas (1992)
	 attained a uniform strong consistency rate of $O(n^{-1/4} (\log \log n)^{1/4})$. Furthermore, Wu et al. (2025) achieved a rate of $O(n^{-1/4} (\log n)^{1/4})$ under AANA sequences. Both of these rates are slower than the result in \eqref{eqn-11}. Consequently, Theorem \ref{th:2} generalizes and extends the results of Cai and Roussas (1992) and Wu et al. (2025) from $\alpha$-mixing and AANA sequences to the framework of $m$-LENQD sequences.
\end{remark}

\subsection{Kernel estimator for the mode}
Among the statistics used to characterize features of a population distribution, in addition to the mean and quantiles, the mode is also an important quantity of interest. It refers to the value at which a distribution exhibits a pronounced concentration and  represents the typical level of the data from an alternative perspective. In some contexts, the mode may better capture the intuitive notion of a typical value in a population than the mean or median. Moreover, it is more robust to extreme observations and is easy to compute and extract in practice, and therefore receives particula attention in many fields. In particular, in the problem of estimating the mode of a density function, the mode is often used in unsupervised learning to measure data typicality and, in cluster analysis, to characterize cluster centers. Hence, studying the mode is of substantial practical significance. In this paper, under the assumption that the probability density function \(f(x)\) is unimodal, the population mode $\theta$ is defined as
\begin{equation}
	\theta := \arg\sup_{x\in\mathbb{R}} f(x).\label{eqn-1}
\end{equation}

As to mode estimation, Parzen (1962) established the weak consistency and asymptotic normality of the kernel mode estimator in the i.i.d. setting. Following Parzen’s work, mode estimation for independent samples has been extensively studied. See, for example,  Konakov (1973) and Romano (1988). In recent years, substantial progress has been made on mode estimation under dependent data. For instance,  Khardani et al. (2012) established strong consistency and convergence rates for the kernel mode estimator under right-censored \(\alpha\)-mixing data.  Hwang and Shin (2016) investigated mode estimation for strictly stationary \(\psi\)-weakly dependent sequences based on kernel density estimators. Wang, W. et al. (2025) studied the consistency of the mode estimator \(\theta_n\) for $m$-END samples when \(f_n(x)\) is a frequency polygon estimator. Motivated by these works, we further examine the strong consistency of kernel mode estimation for $m$-LENQD sequences.

A natural estimator of the population mode $\theta$ can be derived from $\hat f_n(x)$ given in \eqref{eq-2} as follows

\begin{equation}\label{eqn-21}
	\theta_n \in \arg \sup_{x \in \mathbb{R}} \hat{f}_n(x).
\end{equation}

Based on Theorem \ref{th:1}, we can establish the strong consistency for the kernel mode estimator $\theta_n$ as follows.
\begin{theorem}\label{th:3}
	Let $\{X_n, n \geq 1\}$ be a sequence of $m$-LENQD random variables, and let (H1)-(H4) hold. Then, 
	\begin{equation}\label{eqn-12}|\theta_n - \theta| = O\left(\left(\frac{\log n}{n h_n^2}\right)^{1/4}\right) + O(h_n), \quad \text{a.s.}
	\end{equation}
\end{theorem}

\begin{remark}
	In fact, the estimate \(\theta_n\) satisfying \eqref{eqn-21} is not necessarily unique. Therefore, all the results in this paper will pertain to any sequence  of random variables \(\theta_n\) satisfying$f(\theta_n)=\sup_{x\in\mathbb{R}}\hat f_n(x).$
	In particular, we can specify our choice by taking
	$\theta_n=\inf\left\{s\in\mathbb{R}\,\big|\,\hat f_n(s)=\sup_{x\in\mathbb{R}}\hat f_n(x)\right\}.$
\end{remark}
	
\section{Simulation Study}
In this section, we conduct a simulation study using R software to investigate the  finite sample size performance of our results. 
For any fixed $n \ge 3$ and $m \ge 1$, let normal random vector
$(e_1,e_2,\cdots,e_{n+m})^T \sim N_{n+m}(\boldsymbol{0},\boldsymbol{\varSigma_1})$, where \(\mathbf{0}\) denotes the zero vector and the covariance matrix
\[
\varSigma_1 = 
\begin{pmatrix}
	1+\rho^2 & -\rho & 0 & \cdots & 0 & 0 & 0 \\
	-\rho & 1+\rho^2 & -\rho & \cdots & 0 & 0 & 0 \\
	0 & -\rho & 1+\rho^2 & \cdots & 0 & 0 & 0 \\
	\vdots & \vdots & \vdots & & \vdots & \vdots & \vdots \\
	0 & 0 & 0 & \cdots & 1+\rho^2 & -\rho & 0 \\
	0 & 0 & 0 & \cdots & -\rho & 1+\rho^2 & -\rho \\
	0 & 0 & 0 & \cdots & 0 & -\rho & 1+\rho^2 \\
\end{pmatrix}_{(n+m) \times (n+m)},
\]
with \(0<\rho<1\). As established by Joag‑Dev and Proschan (1983), multivariate normal random vectors with non‑positive off‑diagonal covariance entries are NA. It follows that $(e_1,e_2,\dots,e_{n+m})$ is a NA vector, and the resulting random variables represent a special case of $m$-LENQD sequences.
\begin{figure}[!ht]
	\centering
	\includegraphics[width=0.9\textwidth]{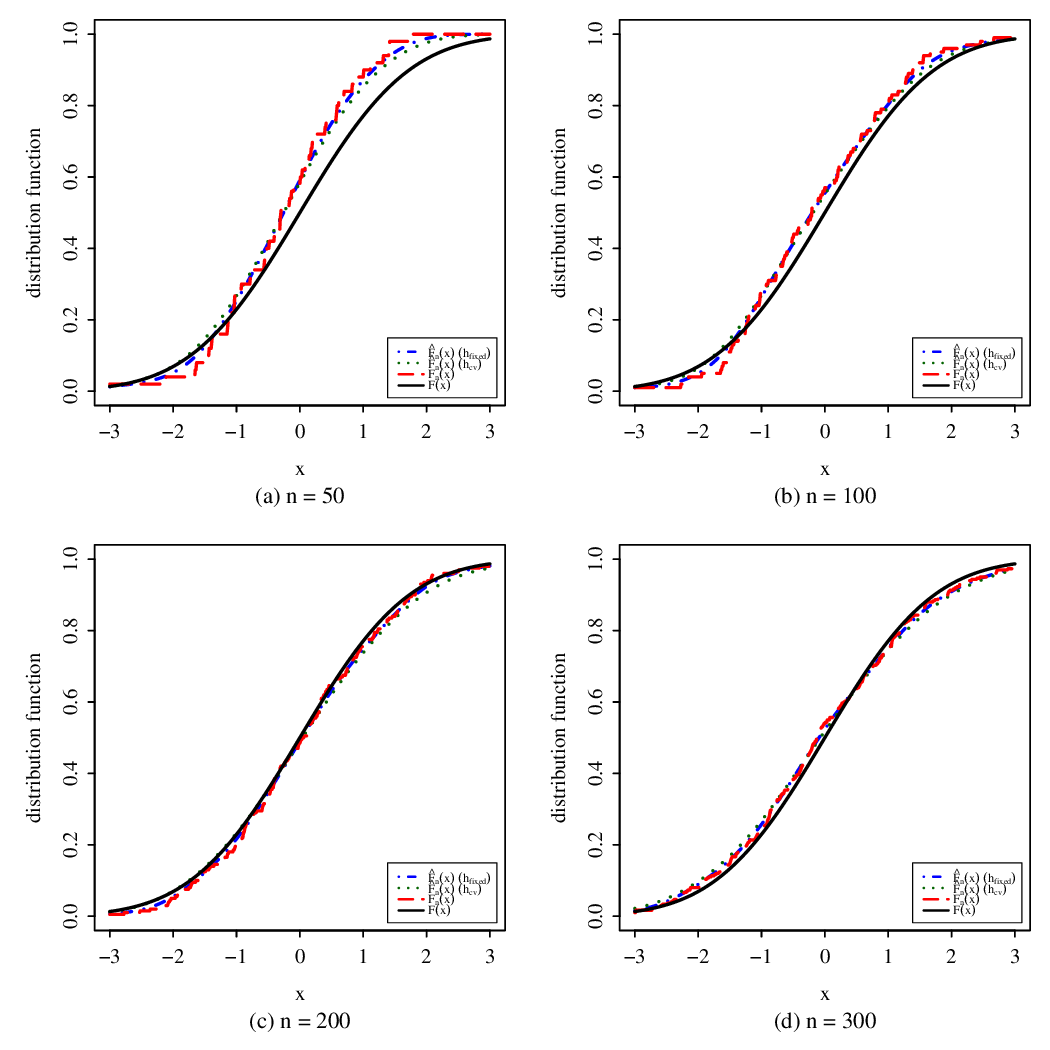}
	\caption{Comparison of histogram density estimator and kernel function estimator with $\rho = 0.1$. (a) $n=50$, (b) $n=100$, (c) $n=200$, and (d) $n=300$.}\label{fig:1}
\end{figure}
For fixed positive integer $m$, let $X_i=\sum_{k=1}^m e_{i+k-1}$ for each $i\ge 1$, then it is easy to show that $\{X_i,i\ge 1\}$ is a sequence of m‑LENQD random variables.

\begin{figure}[!ht]
	\centering
	\includegraphics[width=0.9\textwidth]{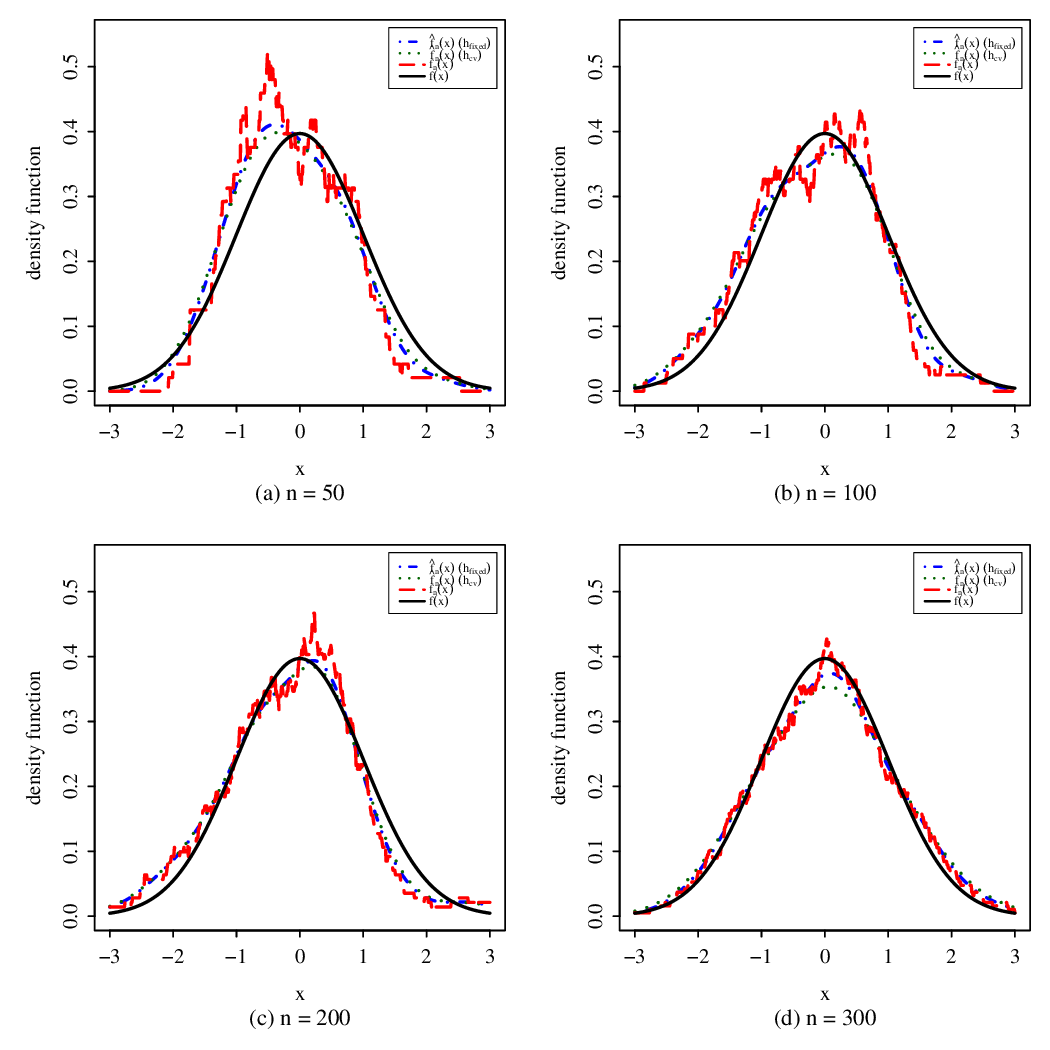}
	\caption{Comparison of histogram density estimator and kernel density estimator with $\rho = 0.1$. (a) $n=50$, (b) $n=100$, (c) $n=200$, and (d) $n=300$.}\label{fig:2}
\end{figure}

In the simulation study, we consider sample sizes $n = 50, 100, 200, 300$ and set $\rho = 0.1$, $m=3$. In fitting the kernel-type estimators, we use a Gaussian kernel and consider two bandwidth choices: a fixed bandwidth $h_{\text{fixed}} = n^{-1/5}$ and a cross-validation (CV) bandwidth $h_{\text{cv}}$.
To highlight the advantage of kernel-type estimators in terms of smoothness, we compare $\hat{F}_n(x)$ with $F_n(x)$ and $\hat{f}_n(x)$ with $f_n(x)$. Figures \ref{fig:1}--\ref{fig:3} illustrate the performance of the fitted curves across different sample sizes.

Specifically, Figure \ref{fig:1} presents the estimation results for  $\widehat{F}_n(x)$ and $F_n(x)$. It can be drawn from Figure \ref{fig:1} that: (i) the fitted curves of $F_n(x)$ are less smooth than those of $\widehat{F}_n(x)$ for a fixed sample size; (ii)  both $F_n(x)$ and $\widehat{F}_n(x)$ converge to the true distribution function $F(x)$ as the sample size $n$ increases; (iii) there is no significant difference in the performance of the kernel distribution estimator between $h_{\text{fixed}}$ and $h_{\text{cv}}$. Furthermore, similar conclusions can be drawn from the performance of $\hat{f}_n(x)$ versus $f_n(x)$ in Figure \ref{fig:2}, and of $\hat{r}_n(x)$ versus $r_n(x)$ in Figure \ref{fig:3}.

\begin{figure}[!ht]
	\centering
	\includegraphics[width=0.9\textwidth]{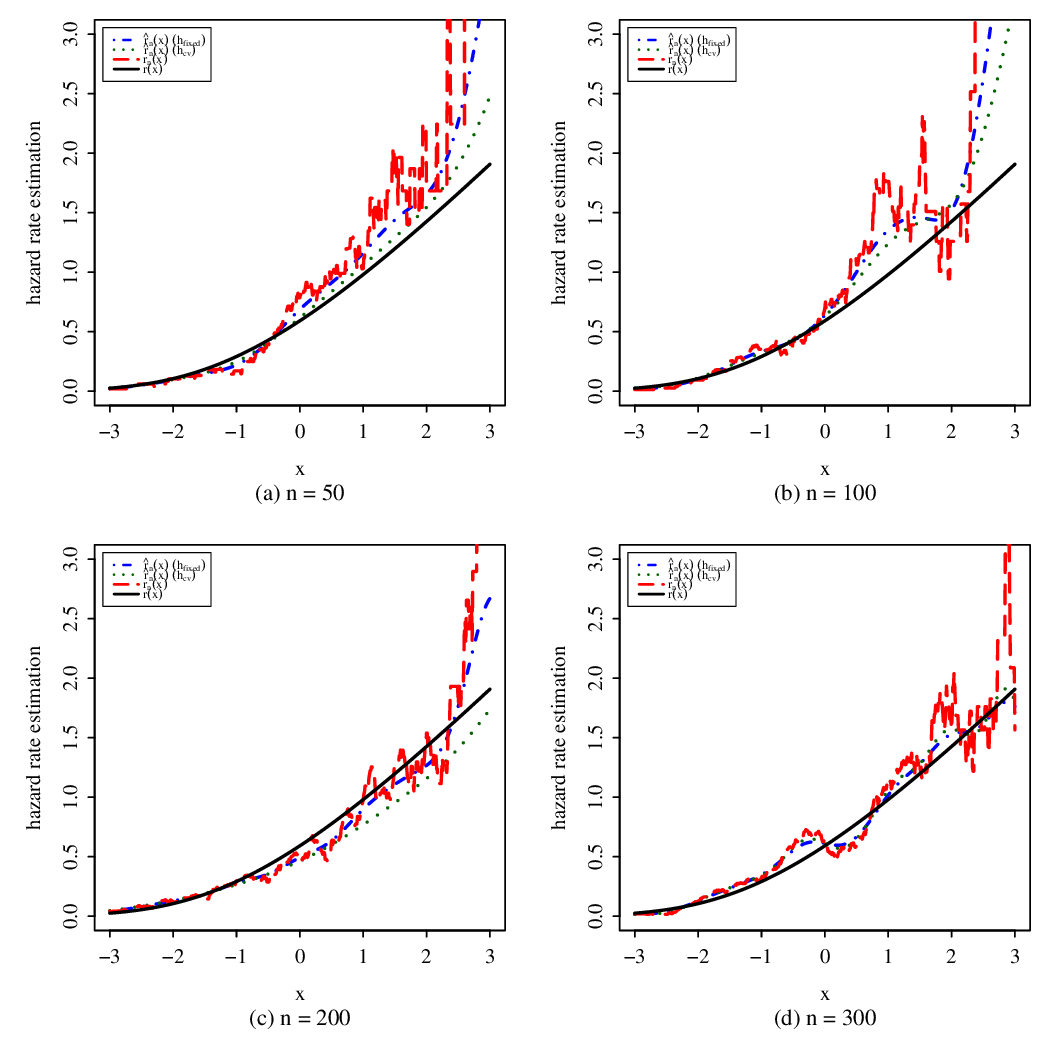}
	\caption{Comparison of histogram density estimator and hazard rate estimator with $\rho = 0.1$. (a) $n=50$, (b) $n=100$, (c) $n=200$, and (d) $n=300$.}\label{fig:3}
\end{figure}

Additionally, we conduct  Monte Carlo simulations with $M=1000$ replications to compute the global mean squared error (GMSE)  of the estimators $\hat{F}_n(x)$, $\hat{f}_n(x)$, and $\hat{r}_n(x)$, as well as their GMSE ratios relative to $F_n(x)$, $f_n(x)$, and $r_n(x)$. The GMSEs of $\hat{F}_n(x)$ and $F_n(x)$ are defined as follows:
\[
\mathrm{GMSE}(\hat{F}_n) = \frac{1}{MH} \sum_{k=1}^{M} \sum_{j=1}^{H}
\left(\hat{F}_{n,k}(x_j) - F(x_j)\right)^2,
\]
and
\[
\mathrm{GMSE}(F_n) = \frac{1}{MH} \sum_{k=1}^{M} \sum_{j=1}^{H}
\left(F_{n,k}(x_j) - F(x_j)\right)^2,
\]
respectively, where $H$ is the number of equidistant grid points $x_j$ in the interval $[-3, 3]$, and $\hat{F}_{n,k}(x_j)$ (or $F_{n,k}(x_j)$) denotes the estimate computed at iteration $k$. The GMSEs of $\hat{f}_n(x)$ and $f_n(x)$ are defined analogously. For the estimators $\hat{r}_n(x)$ and $r_n(x)$, the calculation interval is restricted to $[-3, 1]$ to avoid numerical instability caused by denominators approaching zero. The corresponding numerical results are summarized in Table \ref{table:1}. 

\begin{figure}[!ht]
	\centering
	\includegraphics[width=0.9\textwidth]{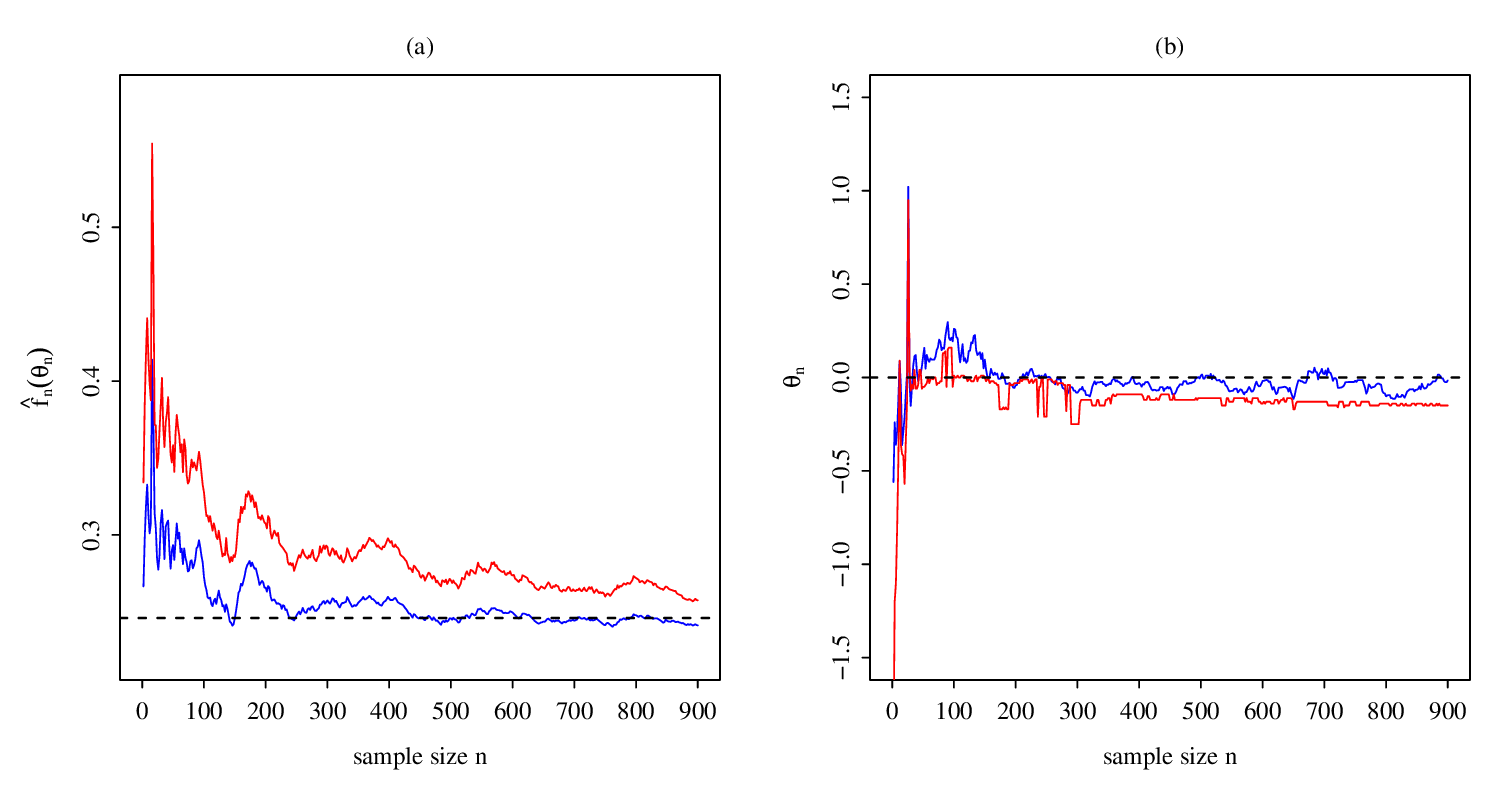}
	\captionsetup{justification=raggedright, singlelinecheck=false}
	\caption{Typical trajectories of $\max \hat{f}_n(x)$ and $\theta_n$ for estimating the mode($\theta$=0).}\label{fig:4}
\end{figure}

As shown in Table \ref{table:1}, the GMSE values of $\hat{F}_n(x)$, $\hat{f}_n(x)$, and $\hat{r}_n(x)$ decrease as the sample size $n$ increases. Moreover, the GMSE ratios between the histogram-type and kernel-type estimators (e.g., $\mathrm{GMSE}(F_n)/\mathrm{GMSE}(\hat{F}_n)$) are greater than 1 in the vast majority of cases, with only a few exceptions being slightly less than 1 but remains very close to it. This indicates that the estimation error of kernel-type estimators is generally smaller than that of their histogram-type counterparts. These results are in good agreement with the conclusions drawn from Figures \ref{fig:1}--\ref{fig:3}. Consequently, we can conclude that the kernel-type estimators $\hat{F}_n(x)$, $\hat{f}_n(x)$, and $\hat{r}_n(x)$ exhibit substantially better estimation performance than the histogram‑type estimators $F_n(x)$, $f_n(x)$, and $r_n(x)$.

Next, to illustrate the consistency of the mode estimator established in Theorem \ref{th:3}, we conduct a simulation for sample sizes $n \in \{2, 3, \dots, 900\}$. The bandwidth $h_{\mathrm{cv}}$ is adopted, while other settings remain the same as before. In Figure \ref{fig:4}, the red lines represent the histogram density estimator, whereas the blue lines denote the kernel density estimator. The figure shows that, as the sample size $n$ increases, the estimated density peak converges toward the true value $f(0)$, while the mode estimators vary around the true mode $0$. Specifically, Figure \ref{fig:4}(a) shows that  $\max_{x\in\mathbb{R}} \hat{f}_n(x)$ and  $\max_{x\in\mathbb{R}} f_n(x)$ rapidly approach $f(0)$ as $n$ increases from 2 to 900, and it is obvious that $X_i \sim N(0, 2.63)$, the density at the true mode is $f(0) = \frac{1}{\sqrt{2 \times 2.63 \pi}} \approx 0.245998$. Figure \ref{fig:4}(b) displays how the two mode estimators converge to the ture mode. It is easily observed that the kernel mode estimator is closer to the true value than its histogra$m$-type counterpart. Overall, Figure \ref{fig:4}(a,b) reveal that the convergence performance of the kernel mode estimator is superior to that of the histogram estimator. These results are in good agreement with the theoretical consistency results of the mode estimator established in Theorem \ref{th:3}.

\begin{table}[ht]
	\caption{Simulation results of GMSE for different estimators}\label{table:1}%
	\centering
	\small
	\setlength{\tabcolsep}{3.5pt}
	\begin{tabular}{@{}ccrrrrrr@{}}
		\toprule
		\multicolumn{2}{c}{} &
		\multicolumn{2}{c}{Distribution estimator} &
		\multicolumn{2}{c}{Density estimator} &
		\multicolumn{2}{c}{Hazard rate estimator} \\
		\cmidrule(lr){3-4}\cmidrule(lr){5-6}\cmidrule(lr){7-8}
		$h_n$ & $n$ &
		$\mathrm{GMSE}(\hat{F}_n)$ &
		$\mathrm{GMSE}\!\left(\frac{F_n}{\hat{F}_n}\right)$ &
		$\mathrm{GMSE}(\hat{f}_n)$ &
		$\mathrm{GMSE}\!\left(\frac{f_n}{\hat{f}_n}\right)$ &
		$\mathrm{GMSE}(\hat{r}_n)$ &
		$\mathrm{GMSE}\!\left(\frac{r_n}{\hat{r}_n}\right)$ \\
		\midrule
		$h_{\mathrm{fixed}}$ & 50  & 0.003934 & 1.159637 & 0.001658 & 1.718394 & 0.013906 & 1.865313 \\
		& 100 & 0.002029 & 1.130425 & 0.000949 & 1.714759 & 0.007152 & 1.730329 \\
		& 200 & 0.001010 & 1.109248 & 0.000548 & 1.717313 & 0.003923 & 1.672252 \\
		& 300 & 0.000761 & 1.086141 & 0.000405 & 1.685564 & 0.002868 & 1.608616 \\
		\midrule
		$h_{\mathrm{cv}}$ & 50  & 0.004187 & 1.089483 & 0.001656 & 1.525437 & 0.013623 & 1.789694 \\
		& 100 & 0.002181 & 1.051950 & 0.000932 & 1.524854 & 0.007144 & 1.574964 \\
		& 200 & 0.001106 & 1.012650 & 0.000523 & 1.496469 & 0.003969 & 1.448807 \\
		& 300 & 0.000828 & 0.997646 & 0.000367 & 1.453299 & 0.002848 & 1.357993 \\
		\bottomrule
	\end{tabular}
\end{table}

\section{Real data analysis}

To demonstrate the practical relevance of this study, we apply our proposed estimators to the \texttt{gtemp\_ocean} dataset (Global Mean Ocean Surface Temperature Deviations, 1850--2023) for an empirical analysis. This dataset, available in the \texttt{astsa} package in R, records the annual sea surface temperature anomalies of the global open ocean from 1850 to 2023. The data are measured in degrees Celsius, and all anomalies are calculated relative to the 1991--2020 average temperature baseline. Specifically, the mean temperature over the 1991–2020 period is subtracted from each observation, so that the anomaly series has a mean of zero;  consequently, positive and negative values indicate deviations above and below this reference level, respectively.

First, a preliminary analysis of the \texttt{gtemp\_ocean} dataset is conducted. The time series plot, along with its autocorrelation function (ACF) and partial autocorrelation function (PACF), is displayed in Figure \ref{fig:5}. As observed from the time series plot in Figure \ref{fig:5}(a), the sequence exhibits a significant linear upward trend beginning around 1900, indicating that the series is not stationary.

\begin{figure}[!ht]
	\centering
	\includegraphics[width=0.9\textwidth]{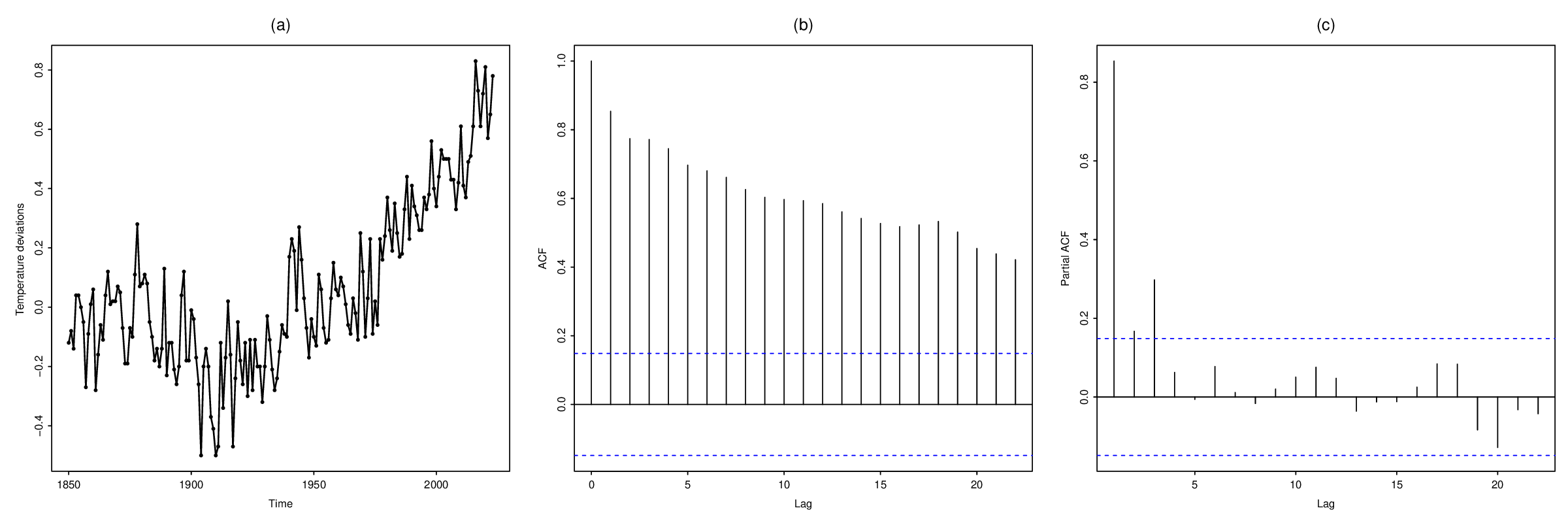}
	\caption{(a) Times series of global mean ocean surface temperature deviations from 1850 to 2023, (b) Autocorrelation function, (c) Partial autocorrelation function.}\label{fig:5}
\end{figure}

\begin{figure}[!ht]
	\centering
	\includegraphics[width=0.9\textwidth]{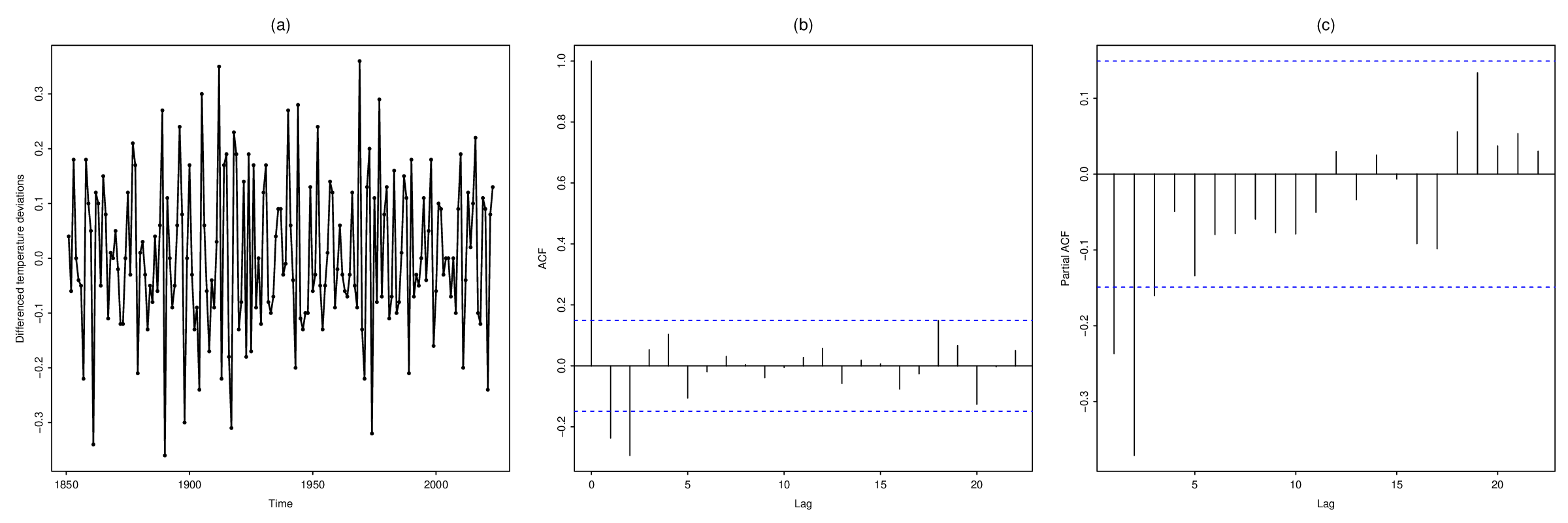}
	\caption{(a) Differences ($y_t$) of Basel yearly temperature series, (b) Autocorrelation function, (c) Partial autocorrelation function.}\label{fig:6}
\end{figure}
\begin{figure}[!ht]
	\centering
	\includegraphics[width=0.9\textwidth]{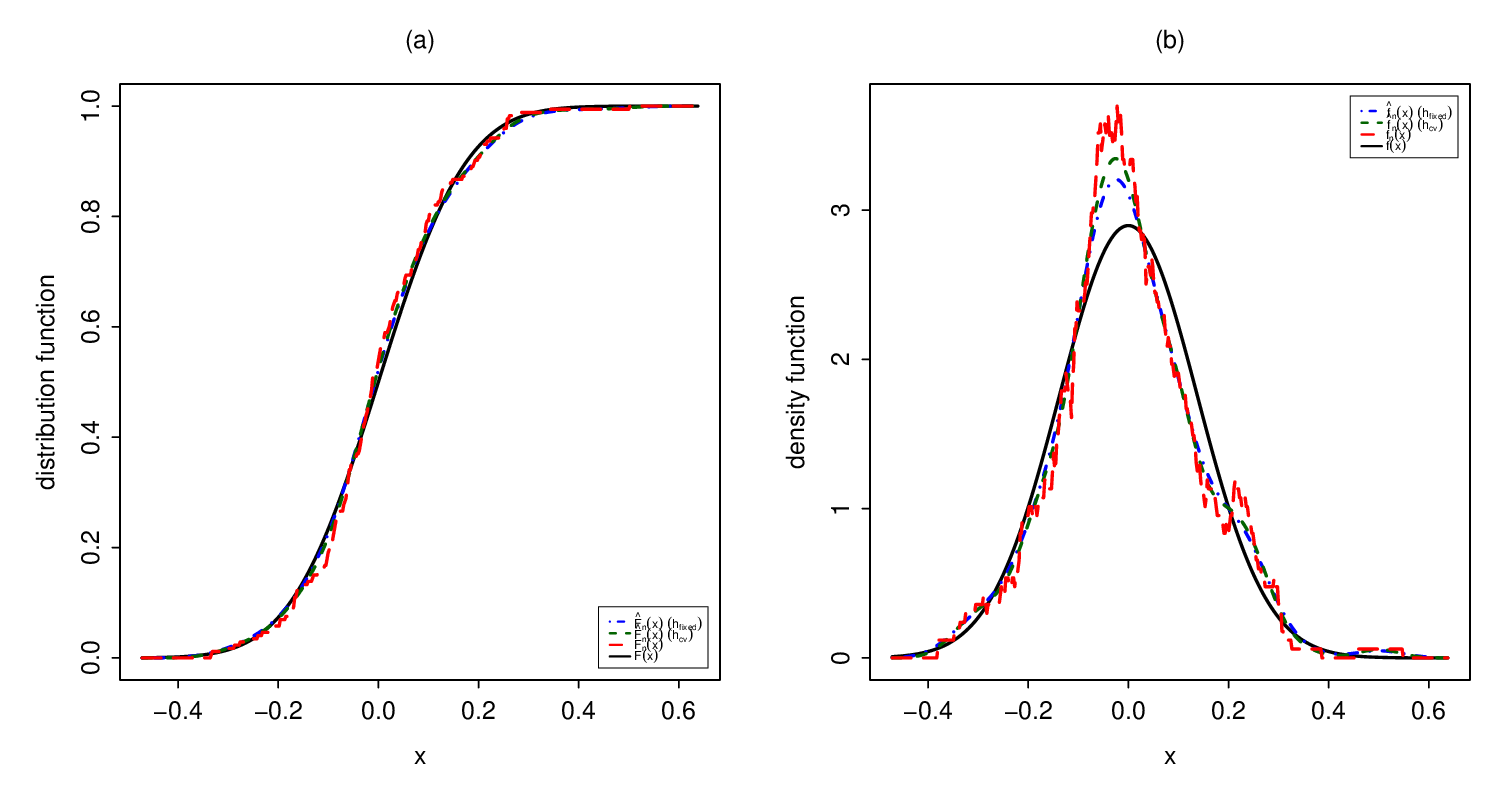}
	\caption{(a) Distribution function estimators and (b) density function estimators for the \texttt{gtemp\_ocean} dataset.}\label{fig:7}
\end{figure}

Based on this observation, to ensure the accuracy and reliability of the model fitting, we apply first-order differencing to the original sequence to construct the differenced series $y_t = x_t - x_{t-1}$. The time series plot, ACF, and PACF of the differenced sequence are shown in Figure \ref{fig:6}. It is evident that the sequence achieves stationarity after first-order differencing. To further verify this, the Augmented Dickey-Fuller (ADF) test is employed, and the null hypothesis of a unit root is rejected, confirming the stationarity of the series. Subsequently, the trend of the differenced series $y_t$ and its ACF and PACF in Figure \ref{fig:6} indicate that the differenced process is stationary. Accordingly, the moving average model of order 2, MA(2), is considered for fitting the differenced data 
$
	y_t = Z_t + \theta_1 Z_{t-1} + \theta_2 Z_{t-2},	
$
where $\{Z_t, t \geq 1\}$ is a Gaussian white noise sequence. The estimated parameters are $\hat{\theta}_1 = -0.4066$, $\hat{\theta}_2 = -0.2934$, $\hat{\mu}_Z = 0$, and $\hat{\sigma}_Z^2 = 0.01516$. Moreover, after comparing several candidate models, including the autoregressive model of order 3, AR(3), the MA(2) model is selected as the optimal model according to the Akaike information criterion (AIC).

To validate the model, the residuals of the MA(2) model are firstly tested for normality using the Jarque–Bera test, which yields a $p$‑value of 0.8585, suggesting that the residuals follow a normal distribution. Subsequently, the Ljung–Box test is employed to examine the autocorrelation of the residuals. Since the $p$‑values at various lags are all greater than 0.05, we fail to reject the null hypothesis of “no autocorrelation” at the 5\% significance level, indicating that the residual sequence is White noise. These tests confirm the effectiveness of the MA(2) model for fitting the differenced series. Since linear transformations or combinations of multivariate normal variables remain multivariate normal, difference of real data follows a multivariate normal distribution with mean vector zero and a covariance matrix $\varSigma_2$ given by
\[\varSigma_2 =
\begin{pmatrix}
	\rho_0\sigma_Z^2 & \rho_1\sigma_Z^2 & \rho_2\sigma_Z^2 & \cdots & 0 & 0 & 0 \\
	\rho_1\sigma_Z^2 & \rho_0\sigma_Z^2 & \rho_1\sigma_Z^2 & \cdots & 0 & 0 & 0 \\
	\rho_2\sigma_Z^2 & \rho_1\sigma_Z^2 & \rho_0\sigma_Z^2 & \cdots & 0 & 0 & 0 \\
	\vdots & \vdots & \vdots & & \vdots & \vdots & \vdots \\
	0 & 0 & 0 & \cdots & \rho_0\sigma_Z^2 & \rho_1\sigma_Z^2 & \rho_2\sigma_Z^2 \\
	0 & 0 & 0 & \cdots & \rho_1\sigma_Z^2 & \rho_0\sigma_Z^2 & \rho_1\sigma_Z^2 \\
	0 & 0 & 0 & \cdots & \rho_2\sigma_Z^2 & \rho_1\sigma_Z^2 & \rho_0\sigma_Z^2
\end{pmatrix}_{n \times n},
\]
where $\rho_0 = 1 + \theta_1^2 + \theta_2^2$, $\rho_1 = \theta_1(1 + \theta_2)$, and $\rho_2 = \theta_2$. According to the result of Joag-Dev and Proschan (1983), a multivariate normal distribution is NA if the off-diagonal elements of its covariance matrix are non-positive. Given $\hat{\theta}_1 = -0.4066$ and $\hat{\theta}_2 = -0.2934$, the resulting sequence is NA. Therefore, the differenced \texttt{gtemp\_ocean} data satisfy the $m$-LENQD property, as NA sequence is a special case of $m$-LENQD sequence.

Finally, using both the fixed bandwidth $h_{\text{fixed}}$ and the cross-validation bandwidth $h_{\text{cv}}$, we apply a Gaussian kernel to estimate the distribution and density function of the \texttt{gtemp\_ocean} data. For comparison, the histogram method is also included. Figure \ref{fig:7} shows that both methods fit the distribution function comparably, whereas the kernel method performs better in density estimation. For the \texttt{gtemp\_ocean} dataset, the proposed kernel-type estimators yield smoother fits and better capture the peak structure of the anomaly distribution than the histogram estimator.


\section{Conclusion}
In this paper, we study the uniform strong convergence rates of the kernel distribution estimator $\hat F_n(x)$ and the kernel density estimator $\hat f_n(x)$ for $m$-LENQD random variables under mild conditions. As applications, we further obtain the uniform strong convergence rate of the hazard rate estimator $\hat r_n(x)$ and the strong consistency rate of the mode estimator $\theta_n$. These results sharpen related complete and uniform convergence results in the existing literature.

Numerical simulations confirm the theoretical findings. Compared with histogram-type estimators, the proposed kernel-type estimators show clear advantages in accuracy and smoothness, especially for estimating the hazard rate function $r(x)$. A real data example further illustrates that the proposed estimators provide smoother and more satisfactory fits to the underlying distribution and density. 

The results obtained in this paper can be extended to other nonparametric density estimation models for dependent sequences, such as frequency polygon estimators and nearest‑neighbor estimators. A further direction for future research would be to investigate the convergence properties of kernel density estimators for \(m\)-LENQD sequences under censored or truncated data, like in Shi et al. (2024) and El Alem et al. (2026).  

\section*{Disclosure Statement}
	
No potential conflict of interest was reported by the author(s).
	
\section*{Funding}
	
This work was supported by the Natural Science Foundation of Fujian Province (No. 2025J01969) and by the Key Science and Technology Plan Project of Zhangzhou City (No. zz20250012) and Fujian Key Laboratory of Granular Computing and Applications (No.2025L0704), China.

	
\section*{References}
\begin{description} \itemsep=-\parsep \itemindent=-1.3cm	
\item[ ]Cai, Z., and G. G. Roussas. 1992. Uniform strong estimation under $\alpha$-mixing, with rates. \textit{Statistics \& Probability Letters} 15 (1):47--55.

\item[ ]Cai, Z. 1993. Uniform strong convergence and rates for the kernel estimators of a distribution function and a regression function under weakly dependent assumptions. \textit{Chinese Journal of Applied Probability and Statistics} 9 (1):11--17.

\item[ ]Cui, Y. J., S. C. Yang, and D. Liang. 2012. Consistency of nearest neighbor estimation of density function for linearly negative quadrant dependent samples. \textit{Journal of Guangxi Normal University (Natural Science Edition)} 30 (2):59--65.

\item[ ]El Alem, M. K., Z. Guessoum, and A. Tatachak. 2026. On kernel mode estimation under RLT and WOD model. \textit{Communications in Statistics-Theory and Methods} 55 (10):2952--70.

\item[ ]El Alem, M. K. 2026. Some results on $m$-linearly extended negative quadrant dependent random variables. \textit{Communications in Statistics-Theory and Methods} 55 (2):434--46.

\item[ ]Hwang, E., and D. W. Shin. 2016. Kernel estimators of mode under $\psi$-weak dependence. \textit{Annals of the Institute of Statistical Mathematics} 68 (2):301--27.

\item[ ]Joag‑Dev, K., and F. Proschan. 1983. Negative association of random variables with applications. \textit{The Annals of Statistics} 11 (1):286--8.

\item[ ]Konakov, V. D. 1973. On asymptotic normality of the sample mode of multivariate distributions. \textit{Teoriya Veroyatnostei i ee Primeneniya} 18 (4):836--42.

\item[ ]Khardani, S., M. Lemdani, and E. Ould Saïd. 2012. On the strong uniform consistency of the mode estimator for censored time series. \textit{Metrika} 75 (2):229--41.

\item[ ]Li, Y. M. 2017. On the rate of strong convergence for a recursive probability density estimator of END samples and its applications. \textit{J. Math. Inequal.} 11 (2):335--43.

\item[ ]Ling, N. X., C. M. Xu, and X. Z. Peng. 2008. Consistency of the kernel estimation of the density function for pairwise negative quadrant dependent sequences. \textit{Journal of Hefei University of Technology (Natural Science Edition)} 31 (2):287--90.

\item[ ]Li, Y. M., and S. C. Yang. 2005. Strong convergence rate of recursive probability density estimator for NA sequences. \textit{Chinese Journal of Engineering Mathematics} 22 (4):659--65.

\item[ ]Liu, R., and L. Yang. 2008. Kernel estimation of multivariate cumulative distribution function. \textit{Journal of Nonparametric Statistics} 20 (8):661--77.

\item[ ]Li, Y. M., Y. Zhou, and C. Liu. 2018. On the convergence rates of kernel estimator and hazard estimator for widely dependent samples. \textit{Journal of Inequalities and Applications} 2018 (1):71.

\item[ ]Li, Y. M., W. C. Pang, Z. Q. Feng, and N. Li. 2023. On the linearly extended negative quadrant dependent random variables and its inequalities. \textit{Communications in Statistics-Theory and Methods} 52 (24):8696--711.

\item[ ]Li, Y. M., W. C. Pang, Z. Q. Feng, and N. Li. 2024. Asymptotic properties of Kaplan–Meier estimator and hazard estimator for censored survival time with LENQD data. \textit{Statistical Theory and Related Fields} 8 (2):107--16.

\item[ ]Parzen, E. 1962. On estimation of a probability density function and mode. \textit{The Annals of Mathematical Statistics} 33 (3):1065--76.

\item[ ]Romano, J. P. 1988. On weak convergence and optimality of kernel density estimates of the mode. \textit{The Annals of Statistics} 16 (2):629--47.

\item[ ]Roussas, G. G. 1995. Asymptotic normality of a smooth estimate of a random field distribution function under association. \textit{Statistics \& Probability Letters} 24 (1):77--90.

\item[ ]Shi, J. H., J. S. Xu, and J. F. Xu. 2024. Strong asymptotic properties of kernel smoothing estimation for NA random variables with right censoring. \textit{Communications in Statistics-Theory and Methods} 53 (12):4531--41.

\item[ ]Wang, W., Q. Q. Wu, and X. F. Tang. 2022. Kernel density estimation under widely orthant dependence. \textit{Journal of University of Science and Technology of China} 52 (2):2.

\item[ ]Wang, W., H. Wang, and Y. Wu. 2025. Consistency of the frequency polygon estimators of density function and density mode under $m$-END samples. \textit{Communications in Statistics-Theory and Methods} 54 (5):1271--86.

\item[]Wang, S., M. M. Wang, W.  T. Wan, D. Yuan, and X. Q. Zheng. 2025. Convergence properties of the kernel‑type density estimator under $m$-WOD dependent samples. \textit{Communications in Statistics‑Simulation and Computation} 1--13. 

\item[ ]Wu, S. P., Y. Wu, W. Z. Yang, and X. J. Wang. 2025. The uniformly strong consistency of kernel‑type distribution estimator under asymptotically almost negatively associated samples. \textit{Statistical Theory and Related Fields} 9 (2):124--40.

\item[ ]Yamato, H. 1973. Uniform convergence of an estimator of a distribution function. \textit{Bulletin of Mathematical Statistics} 15 (3‑4):69--78.

\item[ ]Yang, S. C. 2003. Consistency of nearest neighbor estimator of density function for negative associated samples. \textit{Acta Mathematicae Applicatae Sinica (English Series)} 26 (3):385--94.

\end{description}

\appendix
\section*{Appendix A}

\setcounter{lemma}{0}
\renewcommand{\thelemma}{A.\arabic{lemma}}
\renewcommand{\theequation}{A.\arabic{equation}}

Before presenting the proofs of the main results, we introduce several fundamental lemmas.

\begin{lemma}(El Alem (2026))\label{lem:1}
Let $\{X_n, n\ge 1\}$ be a sequence of $m$-LENQD random variables. If $\{f_n(\cdot), n\ge 1\}$ are all non‑decreasing (or  non‑increasing) functions, then $\{f_n(X_n), n\ge 1\}$ are still $m$-LENQD random variables with the same dominating constant $M$.
\end{lemma}

\begin{lemma}(El Alem (2026))\label{lem:3}
	Let $\{X_n, n \ge 1\}$ be a sequence of $m$-LENQD random variables with $E(X_i) = 0$ and $|X_i| \le b$ for each $i \ge 1$, where $b$ is a positive constant. Then, for each $n$ and any $\epsilon>0$, set $B_n^2 = \sum_{i=1}^n EX_i^2$, there exists a dominating constant $M\ge 1$, such that 
	
	\begin{equation}\label{eqn-12'} 
		P\left(\left|\sum_{i=1}^n X_i\right| > \epsilon\right) \le 2mM \exp\left\{-\frac{(\epsilon/m)^2}{2B_n^2 + (2/3)b(\epsilon/m)}\right\}. 
	\end{equation}
\end{lemma}

\begin{lemma}\label{lem:4}
	Let $\{X_n, n \geq 1\}$ be a sequence of $m$-LENQD random variables with unknown distribution function $F(x)$ and bounded probability density function $f(x)$. Then the empirical distribution function $F_n(x)$ can be constructed by
	\begin{equation}\label{eqn-16} 
		\sup_x |F_n(x) - F(x)| = O\left(\sqrt{\frac{\log n}{n}}\right), \quad \text{a.s.}
	\end{equation}
\end{lemma}

\begin{proof}[Proof of Lemma \ref{lem:4}]
	By Lemma 2 of Yang(2003) , let $\{x_{n,k}\}$ satisfy $F(x_{n,k}) = k/n$ for $n \ge 3$ and $k = 1, \dots, n-1$, then we have
	\[ \sup_{-\infty < x < +\infty} |F_n(x) - F(x)| \le \max_{1 \le k \le n-1} |F_n(x_{n,k}) - F(x_{n,k})| + \frac{2}{n}. 
	\]
	Set $\tau_n = n^{-1/2} \sqrt{\log n}$. It is easy to see that $n\tau_n \to \infty$. Therefore, for all $n$ large enough and any positive constant $D_0$, we have $2/n < D_0 \tau_n/2$. Hence, we can get that
	\begin{align} 
		P\left(\sup_x |F_n(x) - F(x)| > D_0\tau_n\right) &\le P\left(\max_{1 \le k \le n-1} |F_n(x_{n,k}) - F(x_{n,k})| > D_0 \tau_n/2\right) \notag \\
		&\le \sum_{k=1}^{n-1} P\left(|F_n(x_{n,k}) - F(x_{n,k})| > D_0 \tau_n/2\right).\label{eqn-17} 
	\end{align}
	Set $\xi_{j,k} = I(X_j < x_{n,k}) - EI(X_j < x_{n,k})$ for $j = 1, \dots, n$. Then it follows from Lemma \ref{lem:1} that $\{\xi_{j,k}\}$ is still $m$-LENQD with $E\xi_{j,k} = 0$, $|\xi_{j,k}| \le 2$, and $E\bigl(\xi_{j,k}\bigr)^2 \le 1$. Since $\tau_n\to 0$ as $n \to \infty$, we have $D_0 \tau_n/m < 3/4$ for sufficiently large $n$. So applying Lemma \ref{lem:3}, we obtain that
	\begin{align*}
		P\left(|F_n(x_{n,k})-F(x_{n,k})|>D_0 \tau_n/2\right)
		&= P\left(\left|\frac{1}{n}\sum_{j=1}^{n}\xi_{j,k}\right|>D_0\tau_n/2\right) \\
		&\le 2mM \exp\left\{ -\frac{(D_0 n\tau_n)^2 /4m^2}{2n + 4/3 D_0 n\tau_n/m} \right\} \\
		&\le 2mM \exp\left\{ -\frac{D_0^2 n\tau_n^2 /4m^2}{2 + 4/3 D_0\tau_n/m} \right\} \\ 
		&\le 2mM \exp\left\{ -\frac{D_0^2 n\tau_n^2}{12 m^2} \right\} \\
		&\le 2mM n^{-\frac{D_0^2}{12m^2}} 
		\le C n^{-\frac{D_0^2}{12 m^2}}.
		\label{eqn-18}
	\end{align*}
Taking $D_0 >0$  large enough such that $\dfrac{D_0^2}{12m^2} \geq 3$, then
\begin{equation}\label{eqn-19}
	P\left(\left|F_n(x_{nj}) - F(x_{nj})\right| > \frac{D_0\tau_n}{2}\right) \le C n^{-3}.
\end{equation}
Thus, with \eqref{eqn-17} and \eqref{eqn-19}, we obtain
	\[ P\left(\sup_x |F_n(x) - F(x)| > D_0 \tau_n\right) \le Cn^{-2}. 
	\]
Therefore, by the Borel-Cantelli lemma, Lemma \ref{lem:4} holds.
\end{proof}

\begin{proof}[Proof of Theorem \ref{th:1'}]
	According to the definition of $\widehat{F}_n(x)$ in \eqref{eq-2}, we have
	\begin{align*}
		\widehat{F}_n(x) &= \int_{-\infty}^{+\infty} K\left(\frac{x - y}{h_n}\right) dF_n(y) \\
		&= \left. K\left(\frac{x - y}{h_n}\right)F_n(y) \right|_{-\infty}^{+\infty} - \int_{-\infty}^{+\infty} F_n(y)  dK\left(\frac{x - y}{h_n}\right) \\
		&= \int_{-\infty}^{+\infty} F_n(x - h_n u)k(u) du,
	\end{align*}
	where $F_n(x)$ is defined by \eqref{eq-1}. Moreover,
	\begin{align}\label{eq-19}
		\widehat{F}_n(x) - F(x) &= \int_{-\infty}^{+\infty} [F_n(x - h_n u) - F(x - h_n u)]k(u) du \nonumber \\
		&\quad +\int_{-\infty}^{+\infty} [F(x - h_n u) - F(x)]k(u) du.
	\end{align}
	By Lemma \ref{lem:4}, it follows that
	\begin{equation}
		\sup_x \left|\int_{-\infty}^{+\infty} [F_n(x - h_n u) - F(x - h_n u)]k(u) du\right| = O\left(\sqrt{\frac{\log n}{n}}\right), \quad  \text{a.s.}\label{eq-20}
	\end{equation}	
	Applying a second-order Taylor expansion to $F(x - h_n u)$ at $x$, we have
	\begin{equation}
		F(x - h_n u) - F(x) = -f(x)h_n u + \frac{f'(\xi)}{2}h_n^2 u^2,\label{eq-21}
	\end{equation}	
	where $\xi$ lies between $x$ and $x - h_n u$. Consequently, combining \eqref{eq-21} with Assumptions (H1)(i) and (H3), we obtain
	\begin{align*}
		&\int_{-\infty}^{+\infty} [F(x - h_n u) - F(x)]k(u)\,du \\
		&= -h_n f(x) \int_{-\infty}^{+\infty} uk(u)\,du 
		+ \frac{h_n^2 f'(\xi)}{2} \int_{-\infty}^{+\infty} u^2 k(u)\,du \\
		&= O(h_n^2).
	\end{align*}
Thus,
	\begin{equation}
		\sup_x \left|\int_{-\infty}^{+\infty} [F(x - h_n u) - F(x)]k(u) du\right| = O(h_n^2).\label{eq-22}
	\end{equation}
	By combining \eqref{eq-19}--\eqref{eq-22}, we arrive at
	\begin{equation}
		\sup_x |\widehat{F}_n(x) - F(x)| = O\left(\sqrt{\frac{\log n}{n}}\right) + O(h_n^2), \quad \text{a.s.}\label{eq-23}
	\end{equation}
	This completes the proof of Theorem \ref{th:1'}.
\end{proof}

\begin{proof}[Proof of Theorem \ref{th:1}]
	Applying integration by parts and change of variables, we have
	\begin{align}
		\hat{f}_n(x) - f(x) &= \frac{1}{n h_n} \sum_{i=1}^{n} k\left(\frac{x - X_i}{h_n}\right) - f(x) \nonumber \\
		&= \frac{1}{h_n} \int_{-\infty}^{+\infty} k\left(\frac{x - s}{h_n}\right) dF_n(s) - \frac{1}{h_n} \int_{-\infty}^{+\infty} k\left(\frac{x - s}{h_n}\right) dF(s) \nonumber \\
		&\quad + \frac{1}{h_n} \int_{-\infty}^{+\infty} k\left(\frac{x - s}{h_n}\right) dF(s) - f(x) \int_{-\infty}^{+\infty} k(s) ds\nonumber \\
		&= -\frac{1}{h_n}\int_{-\infty}^{+\infty} k(u) dF_n(x - h_n u) + \frac{1}{h_n}\int_{-\infty}^{+\infty} k(u) dF(x - h_n u) \nonumber \\
		&\quad - \frac{1}{h_n}\int_{-\infty}^{+\infty} k(u) dF(x - h_n u) - f(x) \int_{-\infty}^{+\infty} k(u) du \nonumber \\
		&=- \frac{1}{h_n}\left(k(u)F_n(x - h_n u)\Big|_{-\infty}^{+\infty} - \int_{-\infty}^{+\infty} F_n(x - h_n u) dk(u)\right) \nonumber \\
		&\quad +\frac{1}{h_n} \left(k(u)F(x - h_n u)\Big|_{-\infty}^{+\infty} - \int_{-\infty}^{+\infty} F(x - h_n u) dk(u)\right) \nonumber \\
		&\quad + \int_{-\infty}^{+\infty} k(u) [f(x - h_n u) - f(x)] du \nonumber \\
		&= \frac{1}{h_n}\int_{-\infty}^{+\infty} [F_n(x - h_n u) - F(x - h_n u)] dk(u) \nonumber \\
		&\quad + \int_{-\infty}^{+\infty} k(u) [f(x - h_n u) - f(x)] du. \label{eq-24}
	\end{align}
	From Lemma \ref{lem:4} and the condition $\int_{-\infty}^{+\infty} |k'(u)| du < \infty$ in Assumption (H1)(ii), it follows that
	\begin{align}
		&\sup_{x \in \mathbb{R}} \left|\frac{1}{h_n}\int_{-\infty}^{+\infty} [F_n(x - h_n t) - F(x - h_n t)] dk(t)\right| \nonumber \\
		&\le \frac{1}{h_n}\sup_{(x-h_n u) \in \mathbb{R}} |F_n(x - h_n u) - F(x - h_n u)| \int_{-\infty}^{+\infty} |k'(u)| du \nonumber \\
		&= O\left(\sqrt{\frac{\log n}{nh_n^2}}\right), \quad \text{a.s.}\label{eq-25}
	\end{align}
	By Assumption (H1)(i) and the boundedness of $f''(x)$ in Assumption (H3), we obtain
	\begin{equation}
		\sup_{x \in \mathbb{R}}\left|\int_{-\infty}^{+\infty} k(u) [f(x - h_n u) - f(x)] du\right| \le \frac{h_n^2}{2} \sup_{x \in \mathbb{R}}|f''(\eta)| \int_{-\infty}^{+\infty} k(u) u^2 du = O(h^2_n),\label{eq-26}
	\end{equation}	
	where $\eta$ lies between $x$ and $x - h_n u$.
	
	Therefore, according to \eqref{eq-24}--\eqref{eq-26}, we have
	\[ \sup_{x \in \mathbb{R}} |\hat{f}_n(x) - f(x)| = O\left(\sqrt{\frac{\log n}{nh_n^2}}\right) + O(h_n^2), \quad \text{a.s.} 
	\]
	This completes the proof of Theorem \ref{th:1}.
	
\end{proof}
\begin{proof}[Proof of Theorem \ref{th:2}]
	Let $S(x)=1-F(x)$ , $\widehat{S}_n(x) = 1 - \widehat{F}_n(x)$. By \eqref{eq:hazard_rate_def} and \eqref{eq:4.2},
	\begin{equation}\label{eq-27}
		|\hat r_n(x)-r(x)|\le \frac{S(x)|\hat f_n(x)-f(x)|+f(x)|\widehat F_n(x)-F(x)|}{\widehat{S}_n(x)S(x)}.
	\end{equation}
	For $x\in[c,d]$, since $S(\cdot)$ is nonincreasing, $0<S(d)\le S(x)\le S(c)\le 1$. By Theorem 1,
	\[
	\sup_{x\in[c,d]}|\widehat{S}_n(x)-S(x)|
	=\sup_{x\in[c,d]}|\widehat F_n(x)-F(x)|\to0,\quad \text{a.s.}
	\]
	Hence, for $n$ sufficiently large,
	\[
	\sup_{x\in[c,d]}\widehat{S}_n(x)-S(x)|<\frac{S(d)}{2},\quad \text{a.s.}
	\]
	Therefore, for all $x\in[c,d]$,
	\[
	\widehat{S}_n(x)\ge S(x)-|\widehat{S}_n(x)-S(x)|
	\ge  S(d)-\frac{ S(d)}{2}=\frac{S(d)}{2}>0.
	\]
	Thus, for $n$ sufficiently large, by \eqref{eq-27},
	\[
	\sup_{x\in[c,d]}|\hat r_n(x)-r(x)|
	\le \frac{2}{S(d)}\sup_{x\in[c,d]}|\hat f_n(x)-f(x)|
	+\frac{2\sup_{x\in[c,d]}f(x)}{ S(d)^2}\sup_{x\in[c,d]}|\widehat F_n(x)- F(x)|.
	\]
	Using Theorem \ref{th:1'} and Theorem \ref{th:1}, together with the boundedness of $f(x)$, we obtain
	\[
	\sup_{x\in[c,d]}|\hat r_n(x)-r(x)|
	=O\!\left(\sqrt{\frac{\log n}{nh_n^2}}\right)+O(h_n^2),\quad \text{a.s.}
	\]
	This completes the proof of Theorem \ref{th:2}.
\end{proof}

\begin{proof}[Proof of Theorem \ref{th:3}]
	First, note that
	\begin{align}
		|f(\theta_n) - f(\theta)| &\leq |f(\theta_n) - \hat{f}_n(\theta_n)| + |\hat{f}_n(\theta_n) - f(\theta)| \nonumber\\
		&\leq \sup_{x \in \mathbb{R} } |\hat{f}_n(x) - f(x)| + |\hat{f}_n(\theta_n) - f(\theta)| \nonumber \\
		&\leq 2 \sup_{x \in \mathbb{R} } |\hat{f}_n(x) - f(x)|.\label{eq-28}
	\end{align}
	Then, the strong consistency of the mode estimator $\theta_n$ follows from Theorem \ref{th:1} and Assumption (H4). By performing a second-order Taylor expansion of $f(\theta_n)$ around $\theta$, we have
	\[
	f(\theta_n) - f(\theta) = \frac{1}{2} \left( \theta_n - \theta \right)^2 f^{''}(\theta_n^*),
	\]
	where $\theta_n^*$ lies between $\theta$ and $\theta_n$. According to \eqref{eq-28} and Assumptions (H3) and (H4), we get
	\[
	|\theta_n - \theta| \leq 2 \sqrt{\frac{\sup_{x \in \mathbb{R} } |\hat{f}_n(x) - f(x)|}{|f^{''}(\theta_n^*)|}}.
	\]
	Hence, by Theorem \ref{th:1}, we complete the proof of Theorem \ref{th:3}.		
\end{proof}
	
\end{document}